\documentclass[letterpaper, 10 pt, conference]{ieeeconf}  
\usepackage{amsmath} 
\usepackage{amssymb}  
\usepackage[dvipsnames]{xcolor}

\IEEEoverridecommandlockouts                              
\usepackage{color}
\usepackage{multirow}
\usepackage{graphicx}
\usepackage{subfig}
\usepackage{epstopdf}
\usepackage{lipsum}  
\usepackage{mathtools,caption,float,yfonts,cuted}
\usepackage{balance}
\usepackage{nicefrac}

\newtheorem{definition}{Definition}
\newtheorem{theorem}{Theorem}
\newtheorem{lemma}{Lemma}

\newtheorem{assumption}{Assumption}
\newtheorem{example}{Example}
\newtheorem{stassumption}[assumption]{Standing Assumption}

\newcommand{\argmin}{\mathop{\mathrm{argmin}}\limits}

\newcommand{\Id}{ \textrm{Id}}
\newcommand{\bR} { {\mathbb R}}
\newcommand{\bN} { {\mathbb N}}

\newcommand{\fix} { {\mathrm{fix}}}
\newcommand{\zer} { {\mathrm{zer}}}
\newcommand{\diag} { {\mathrm{diag}}}
\newcommand{\prox}{\textrm{prox}}
\newcommand{\proj}{\textrm{proj}}

\newcommand{\ca}[1]{\mathcal{#1}}

\newcommand{\bld}[1]{\boldsymbol{#1}}
\newcommand{\1}{\bld 1}
\newcommand{\0}{\bld 0}
\newcommand{\col }{\mathrm{col}}

\usepackage{enumerate}

\title{\LARGE \bf
Time-varying constrained proximal type dynamics\\ in multi-agent network games 
}

\author{Carlo Cenedese \and Giuseppe Belgioioso \and Sergio Grammatico \and Ming Cao 
\thanks{C. Cenedese and M. Cao are with the Jan C. Wilems Center for Systems and Control, ENTEG, Faculty of Science and Engineering, University of Groningen, The Netherlands	({\texttt{\{c.cenedese, m.cao\}@rug.nl}})
 G.Belgioioso is with the Control System group, TU Eindhoven, 5600 MB Eindhoven, The Netherlands
(\texttt{g.belgioioso@tue.nl}).
 S. Grammatico is with the Delft Center for Systems and Control, TU Delft, The Netherlands
(\texttt{{s.grammatico@tudelft.nl}}).
This work was partially supported by the EU Project \lq MatchIT' (82203), NWO under research project OMEGA (613.001.702) and P2P-TALES (647.003.003) and by the ERC under research project COSMOS (802348).}
}

\begin{document}

\maketitle
\thispagestyle{empty}
\pagestyle{empty}

\begin{abstract}
In this paper, we study multi-agent network games subject to affine time-varying coupling constraints and a time-varying communication network. We focus on the class of games adopting proximal dynamics and study their convergence to a persistent equilibrium. The assumptions considered to solve the problem are discussed and motivated. We develop an iterative equilibrium seeking algorithm, using only local information, that converges to a special class of game equilibria. Its derivation is motivated by several examples, showing that the original game dynamics fail to converge. Finally, we apply the designed algorithm to solve a constrained consensus problem, validating the theoretical results. \end{abstract}
\section{Introduction}

\subsection{Motivation: Multi-agent decision making over networks}

In multi-agent decision making over networks, all the decision makers, in short, \textit{agents}, share their information only with a selected number of agents. 
In particular, the agents' state (or decision) is the result of a \textit{local decision making} process, e.g. a constrained optimization problem, and a \textit{distributed communication} with the neighboring agents, defined by the communication network. In many problems, the goal of the agents is reaching a collective equilibrium state, where no agent can benefit from changing its state. 
The local interation between the agents is exploited in opinion dynamics to model the evolution of a population's collective opinion as an emerging phenomenon of the local interactions, see \cite{grammatico:2017:opinion_dynamics_are_prox,Ghaderi2014, etesami:basar:15}. Another interesting consequence of the communication structure is that the agents keep their own data private, exchanging  information only with selected agents. This characteristic is of particular interest in, for example, traffic and information networks problems \cite{jaina:walrand:10} or in the charging scheduling of electric vehicles \cite{cenedese:belg:gra:cao:2019:Asynch_ECC,callaway:09}. This class of problems arises also in other applications, e.g., in smart grids \cite{dorfer:simpson-porco:bullo:16, grammatico:18tcns} and sensor network \cite{martinez:bullo:cortes:frazzoli:07}, \cite{simonetto:leus:2014:distributed_ML_sensor_network}.


\subsection{Literature overview: Multi-agent optimization and multi-agent network games}

In this work, we study a particular instance of the problem introduced above, namely a multi-agent network game, where the communication network and the constraints between the agents are both time-varying.
Multi-agent network games arise from the well established field of distributed optimization and equilibrium seeking over networks. In the past years, several results were proposed for optimization problems subject to a time-varying communication network: in   \cite{nedic:2015:distr_opt_over_time_var_directed_graphs} the subgradients of the cost functions are bounded and the communication is described by a strongly connected sequence of directed graphs, while in   \cite{nedic:alshevsky:2017:achieving_geometric_convergence_distr_opt} the cost functions are assumed to be continuously differentiable and a linearly convergent algorithm is designed under the assumption of a time-varying undirected communication network. Another approach, explored in \cite{li:marden:2012:games_for_distributed_optimization}, is to construct a game, whose emerging behavior solves the optimization problems. In this case, the cost functions are differentiable and the communication ruled by an undirected time-varying graph connected over time.

The problem of noncooperative multi-agent games, subject to coupling constrains, was firstly studied in \cite{facchinei:kanzow:07}, under the assummptions of  continuosly differentiable cost functions  and no network structure between the agents. In the past years, several researchers focused on this class of problems providing many results for games over networks, e.g.,  in \cite{yi_pavel:2017:disribiuted_primal_dual_conf,belgioioso:grammatico:2018:proj_grad_algorithm_for_GNE,cenedese:belg:gra:cao:2019:Asynch_ECC} where the communication network is always assumed undirected, while the cost functions are chosen either differentiable or continuously differentiable. Moreover, some authors also focused on the class of noncooperative games over time-varying communication network, in particular on the unconstrained case. 
For example, in \cite{koshal:nedic:shanbhag:16} differentiable and strictly convex cost functions with Lipschitz continuous gradient were considered, where the sequence of time-varying communication networks was repeatedly strongly connected, and the associated adjacency matrices doubly stochastic.

\subsection{Paper contribution}
A complete formulation of multi-agent network games, subject to proximal type dynamics, can be found in \cite{grammatico:18tcns} where the unconstrained case is studied for a time-varying strongly connected communication network, described by a doubly stochastic adjacency matrix. In \cite{cenedese:kawano:grammatico:cao:2018:time_varying_proximal_dynamics,cenedese:2019:arXiv}, the condition on the double stochasticity of the adjacency matrix was relaxed, in the first case by means of a dwell time. Notice that these types of games can also be rephrased as paracontracions; in this framework, the work in \cite{fullmer:morse:2018:common_fixed_point_finite_family_paracontraction} provided convergence for repeatedly jointly connected digraphs. 
Iterative equilibrium seeking algorithms were developed for  constrained multi-agent network  games in \cite{grammatico:18tcns,cenedese:2019:arXiv} under the assumption of a static communication network.

In this work, we aim to address the problem of a constrained multi-agent network games subject to a time-varying communication network. In particular, we first discuss the convergence of the game and motivate the technical assumption needed to ensure the existence of an equilibrium, and then we develop an equilibrium seeking algorithm that achieves global convergence for the game at hand. The main difference with the work in \cite{cenedese:2019:arXiv} is the presence of both time-varying communication network and time-varying constraints, and this generalization leads to several technical challenges, requiring a more involved  convergence analysis. 
%

\section{NOTATION}
\subsection{Basic notation}
The set of real, positive, and non-negative numbers are denoted by $\mathbb{R}$, $\mathbb{R}_{>0}$ and $\mathbb{R}_{\geq 0}$, respectively; $\overline{\mathbb{R}}:=\mathbb{R}\cup \{\infty\}$. The set of natural numbers is denoted by $\mathbb{N}$. For a square matrix $A \in \bR^{n\times n}$, its transpose is denoted by $A^\top$, $[A]_{i}$ denotes the $i$-th row of the matrix, and  $[A]_{ij}$ the element in the $i$-th row and $j$-th column.Also, $A\succ 0 $ ($A\succeq 0 $)  stands for a symmetric and positive definite (semidefinite) matrix, while $>$ ($\geq$) describes an element wise inequality. $A\otimes B$ is the Kronecker product of the matrices $A$ and $B$.   The identity matrix is denoted by~$I_n\in\bR^{n\times n}$, and $\0$ ($\1$) represents the vector/matrix with only $0$ ($1$) elements. For $x_1,\dots,x_N\in\mathbb{R}^n$ and $\ca N=\{1,\dots,N \}$, the collective vector is denoted  as $\boldsymbol{x}:=\mathrm{col}((x_i )_{i\in\ca N})=[x_1^\top,\dots ,x_N^\top ]^\top$ and $\bld x_{-i}:=\col(( x_j )_{j\in\ca N\setminus \{i\}})=[x_1^\top,\dots,x_{i-1}^\top,x_{i+1}^\top,\dots,x_{N}^\top]^\top$. Given the  $N$ operators $A_1,\dots,A_N$, $\diag(A_1,\dots,A_N)$ denotes a block-diagonal operators with $A_1,\dots,A_N$ as diagonal elements. 
The Cartesian product of the sets $\Omega_1, \dots, \Omega_N$ is described by $\prod^N_{i=1} \Omega_i$.  Given two vectors $x,y \in \bR^n$ and a symmetric and positive definite matrix $ Q \succ 0$, the weighted inner product and norm are denoted by $\langle\,  x \, | \, y \,\rangle_{Q}$ and $\lVert x \rVert_{Q}$, respectively; the $Q-$induced matrix norm is denoted by $\lVert A\rVert_{Q}$. A real $n$ dimensional Hilbert space obtained by endowing $\mathcal H=(\bR^n,\lVert \, \cdot \, \rVert)$  with the product $\langle\,  x \, | \, y \,\rangle_{Q}$ is denoted by $\mathcal{H}_Q$.

\subsection{Operator-theoretic notations and definitions}
The identity operator is defined by $\Id(\cdot)$. The indicator function $\iota_\mathcal{C}:\bR^n\rightarrow[0,+\infty]$ of $\mathcal{C}\subseteq \bR^n$ is defined as $\iota_\mathcal{C}(x)=0$ if $x\in\mathcal{C}$; $+\infty$  otherwise.
The set valued mapping $N_{\ca C}:\bR^n\rightrightarrows \bR^n$ stands for the normal cone to the set $\mathcal{C}\subseteq \bR^n$, that is $N_{\ca C}(x)= \{ u\in\bR^n \,|\, \mathrm{sup}\langle \ca C-x,u \rangle\leq 0\}$ if $x \in \ca C$ and  $\varnothing$ otherwise. The graph of a set valued mapping $\ca A:\ca X\rightrightarrows \ca Y$ is $\mathrm{gra}(\ca A):= \{ (x,u)\in \ca X\times \ca Y\, |\, u\in\ca A (x)  \}$. For a function $\phi:\bR^n\rightarrow\overline{\mathbb{R}}$, define $\mathrm{dom}(\phi):=\{x\in\bR^n|f(x)<+\infty\}$ and its subdifferential set-valued mapping, $\partial \phi:\mathrm{dom}(\phi)\rightrightarrows\bR^n$, $\partial \phi(x):=\{ u\in \bR^n | \: \langle y-x|u\rangle+\phi(x)\leq \phi(y)\, , \: \forall y\in\mathrm{dom}(\phi)\}$.  The projection operator over a closed set $S\subseteq \bR^n$ is $\textrm{proj}_S(x):\bR^n\rightarrow S$ and it is defined as $\textrm{proj}_S(x):=\mathrm{argmin}_{y\in S}\lVert y - x \rVert^2$. The proximal operator $\mathrm{prox}_f(x):\bR^n\rightarrow\mathrm{dom}(f)$ is defined by $\mathrm{prox}_f(x):=\mathrm{argmin}_{y\in\bR^n}f(y)+\textstyle\frac{1}{2}\lVert x-y\rVert^2$.
A set valued mapping $\ca F:\bR^n\rightrightarrows \bR^n$ is $\ell$-Lipschitz continuous with $\ell>0$, if $\lVert u-v \rVert \leq \ell \lVert x-y \rVert$ for all $(x,u)\, ,\,(y,v)\in\mathrm{gra}(\ca F)$; $\ca F$ is (strictly) monotone if for all $(x,u),(y,v)\in\mathrm{gra}(\ca F)$ $\langle u-v,x-y\rangle \geq (>)0$ holds, and  maximally monotone if there is no monotone operator with a graph that strictly contains $\mathrm{gra}(\ca F)$; $\ca F$ is $\alpha$-strongly monotone if for all $(x,u),(y,v)\in\mathrm{gra}(\ca F)$ it holds $\langle x-y, u-v\rangle \geq \alpha \lVert x-y \rVert^2$. $\mathrm{J}_{\ca F} :=(\Id+\ca F)^{-1}$ denotes the resolvent mapping of $\ca F$. Let $\textrm{fix}(\ca A):=\{x \in \mathbb R^n | \, x \in \ca F(x)  \}$ and $\textrm{zer}(\ca F):=\{x \in \mathbb R^n | \,0 \in \ca F(x)  \}$ denote the set of fixed points and zeros of $\ca F$, respectively.
The operator $\ca A:\bR^n\rightarrow\bR^n$  is $\eta$-averaged  ($\eta$-AVG) in $\ca H_Q$, with $\eta\in(0,1)$, if $\lVert \ca A(x)-\ca A(y) \rVert^2_Q \leq \lVert x-y\rVert^2_Q-\frac{1-\eta}{\eta}\lVert (\Id-\ca A)(x)-(\Id-\ca A)(y)  \rVert^2_Q$, for all $x,y\in\bR^n$; $\ca A$ is nonexpansive (NE) if $1$-AVG; $\ca A$ is firmly nonexpansive (FNE) if $\frac{1}{2}$-AVG; $\ca A$ is $\beta$-cocoercive if $\beta\ca A$ is $\frac{1}{2}$-AVG (i.e., FNE).
The operator $\ca A$ belongs to the class $\mathfrak{ I}$ in $\ca H_P$ if and only if $\mathrm{dom}(\ca A) = \bR^n$ and for all $y\in\fix(\ca A)$ and $x\in\bR^n$ it holds $\lVert x-\ca A x \rVert_P \leq \langle x-\ca A x, x-y\rangle_P$. Several type of operators belongs to this class, e.g. FNE operators and the resolvent of a maximally monotone operator. We refer to \cite{combettes:01} for more properties of operators of class $\mathfrak{ I}$.
\section{Mathematical setup and problem formulation}
\label{sec:problem_formulation}

\subsection{Mathematical formulation}
\label{subsec:math_formulation}
We consider $N$ players (or agents) taking part in a game.
A constrained network game is defined by three main components: the constraints each players has to satisfy, the cost functions to be minimized and the communication network.

The constraints can be divided in two types: local and coupling.  At every time instant $k\in\bN$, each agent $i\in\ca N\coloneqq \{1,\dots,N\}$ adopts an action (or strategy) $x_i\in\bR^n$ belonging to its \textit{local feasible set} $\Omega_i\subset \bR^n$, i.e., the collection of  those strategies meeting its local constraints. 
We assume that this set is convex and closed. 
\smallskip
\begin{stassumption}[Convexity]
\label{ass:local_constrain}
For every $i\in\mathcal{N}$, the set $\Omega_i\subset \mathbb{R}^n$ is non-empty, compact and convex.
\hfill\QEDopen
\end{stassumption}
\smallskip
The agents are also subject to $M$ time-varying affine and separable \textit{coupling constraints}, that generate an entanglement between the strategy chosen by player $i$ and those of the others.
For an agent $i\in\ca N$, at time instant $k\in\bN$, the time-varying set of strategies satisfying the coupling constraints, given the other agents' strategies $\bld{x}_{-i}$, reads as
\smallskip
\begin{equation*}
\label{eq:X_i_coupling_constr_TV}
\ca X_i(\bld x_{-i},k) \coloneqq \left\{y\in \bR^n \, | \, C_i(k) y +\textstyle{\sum_{\substack{j=1\\j\not =i}}^N}C_j(k) x_j (k) \leq c(k)  \right\}\:\,
\end{equation*}
\smallskip  
where $C_j(k)\in\bR^{M\times n}$ and $c(k)\in\bR^M$.

In the following, we refer to the collective vector $\bld x\coloneqq\col((x_i)_{i\in\ca N})\in\bR^{Nn}$ as the \textit{strategy profile} of the game.  
All the  strategies profiles that satisfy both the local and coupling constraints determine the \textit{collective feasible decision set}, defined  as 
\smallskip
\begin{equation}
\label{eq:collective_feasible_set}
\bld{\ca X}(k) \coloneqq \bld{\Omega} \cap \left\{\bld x \in \bR^{Nn} | \bld{C}(k)\bld {x} \leq c(k)  \right\} 
\end{equation}
\smallskip
where $\bld C(k)\coloneqq [C_1(k),\dots,C_N(k)]\in \bR^{M\times Nn}$ and $\bld \Omega \coloneqq \prod_{i=1}^N \Omega_i$.  
\smallskip
\begin{stassumption}
\label{ass:convex_constr_set}
For all $i\in\ca N$ and $k\in\bN$, the \textit{collective feasible decision set} $\bld{\ca X}(k) $ satisfies Slater's condition. 
\hfill \QEDopen
\end{stassumption}
\smallskip
All the players in the network are assumed myopic and rational, and thus each agent $i\in\ca N$ aims only at minimizing its local cost function $J_i(x_i,z)$. The myopic nature of the agents is reflected in the argument of the cost function that depend only on the current strategies of the players (as we will clarify in the following).
In this work, we assume that the cost function have the proximal structure, as defined next.

\smallskip
\begin{stassumption}[Proximal cost functions]
\label{ass:cost_function}
For all $i\in\mathcal{N}$, the function $J_i:\mathbb{R}^n\times \mathbb{R}^n\rightarrow\overline{\mathbb{R}}$ is defined as 
\begin{equation}
J_i(x_i,z) :=  \bar f_{i}(x_i) +\textstyle{\frac{1}{2}}\lVert x_i-z\rVert^2,
\label{eq:def_costFunction}
\end{equation}
where the function $\bar f_{i} \coloneqq f_{i}+\iota_{\Omega_i}:\mathbb{R}^n\rightarrow\overline{\mathbb{R}}$ is convex and lower semi-continuous. 
\hfill\QEDopen
\end{stassumption}
\smallskip
 The cost function is composed of two parts, $\bar f_i$ is the local part and has a double role: describing the local objective of agent $i$, via $f_i$, and ensuring that the next strategy belongs to $\Omega_i$, through the indicator function $\iota_{\Omega_i}$. The quadratic part of $J_i$ works as a regularization term and penalizes the distance of the local strategy from $z$. It is also responsible for the strict-convexity of $J_i$, even though $\bar f_i$ is only lower semi-continuous, see \cite[Th.~27.23]{bauschke:combettes}. 

Before providing a formal description of the second argument $z$ in the cost function, let us introduce the time-varying communication network adopted by the agents. We assume that, at each time instant $k\in\bN$, it is described by a strongly connected digraph, defined via the couple $(\ca V,A(k))$.  The set $\ca V$ represents the nodes of the graph that are the players in the game, i.e., $\ca V = \ca N$, so this set does not vary over time. The matrix $A(k)$ denotes the adjacency matrix of the digraph, at time $k$, where $a_{i,j}(k)\coloneqq [A(k)]_{ij}$. For every $i,j\in\ca N$,  $a_{i,j}(k)\in[0,1]$ is the weight that agent $i$ assigns to the strategy of agent $j$. If $a_{i,j}(k)=0$, then  agent $i$ does not communicate with agent $j$. The set of all the  neighbors of agent $i$ is defined as $\ca N_i(k)\coloneqq \{ j\,|\, a_{i,j}(k)>0 \}$.  The following assumption formalizes the properties of the adjacency matrix required throughout this work.
\smallskip
\begin{stassumption}[Row stochasticity and self-loops]
\label{ass:row_stoch}
At every time instant $k\in\bN$, the communication graph is strongly connected. The matrix $A(k)=[a_{i,j}(k)]$ is row stochastic, i.e., $a_{i,j}(k) \geq 0$ for all $i,j \in \ca N$, 
and $\sum_{j=1}^N a_{i,j}(k)=1$, for all $i \in \ca N$. Moreover, $A(k)$ has strictly-positive diagonal elements, i.e., $\min_{i\in\ca N}a_{i,i}(k)=:\underline{a}_k>0$ .~\hfill\QEDopen
\end{stassumption}
\smallskip

For each agent  $i\in\ca N$, the term $z$ in~\eqref{eq:def_costFunction} represents an aggregative quantity  defined by 
$$z\coloneqq \textstyle{\sum_{j=1}^N} a_{i,j}(k) x_j(k)\,,$$
and  hence it is the average of the neighbors' strategies, weighted via the adjacency matrix $A(k)$.
So, the actual cost function of agent $i$ at time $k$ is $J_i(x_i,\sum_{j=1}^N a_{i,j}(k) x_j(k))$.

As mentioned before, the agents are considered rational, thus their only objective is to minimize their local cost function, while satisfying the local and coupling constraints. The dynamics describing this behavior are the \textit{myopic best response dynamics}, defined, for each player $i\in\ca N$, as:
\begin{equation}
x_i(k+1) =  \argmin_{y\in \ca X_i(\bld x_{-i},k)} \; J_i \left(y,\textstyle\sum_{j=1}^N a_{i,j}(k) x_j(k) \right) \:.
\label{eq:myopic_BR}
\end{equation} 
The interaction of the $N$ players, using dynamics \eqref{eq:myopic_BR}, can be natuarally formalized as a \textit{ noncooperative network game}, defined, for all $k\in\bN$, as
\begin{align}
\label{eq:game}
\forall i \in \ca N:
\begin{cases}
\textstyle
\argmin_{y \in \mathbb R^n}& 
 f_i(y)  + \frac{1}{2}\left\| y - \sum_{j=1}^N a_{i,j}x_j \right\|^2\\
\:\:\text{ s.t. } &  y \in \Omega_i\cap \ca X_i(\bld x_{-i},k)\,,
\end{cases}
\end{align}
where we omitted the time dependency of $a_{i,j}(k)$ and $x_j(k)$ to ease the notation.

\subsection{Equilibrium concept and convergence}
\label{subsec:equilibrium_concept}

For the game in~\eqref{eq:game}, the concept of  equilibrium point is non trivial. A popular equilibrium notion for constrained game is the, so called, \textit{generalized network equilibrium} (GNWE). Loosely speaking, a profile strategy $\bld{\hat x}$ is a GNWE of the game, if no player $i$ can change its strategy  to another \textit{feasible} one while decreasing $J_i \left(\hat x_i,\textstyle\sum_{j=1}^N a_{i,j} \hat x_j \right)$. Notice that, if $A$ does not have self-loops, GNWE boils down to \textit{generalized Nash equilibrium}, see \cite{facchinei:kanzow:07}.  

This idea of equilibrium  cannot be directly applied to  \eqref{eq:game} and in fact every variation in the communication network generates a different game, with its own set of GNWE. Therefore, the  equilibria in which we are interested are those invariant to the changes  in the communication; they take the name of \textit{persistent} GNWE (p--GNWE).  
\smallskip
\begin{definition}[{persistent GNWE}]   
A collective vector $\bar{\bld x} = \col((\bar x_i)_{i\in \ca N})$ is a persistent GNWE (p--GNWE) for the game \eqref{eq:game},  if there exists some $\overline k>0$, such that for all $i\in\ca N$,  
\begin{equation}
\bar x_i = \bigcap_{k\geq\bar k} \quad \argmin_{y\in\ca X(\bar{\bld{x}}_{-i},k)} J_i \left( y,\textstyle\sum_{j=1}^N a_{i,j}(k) \bar x_j \right) \:.
\label{eq:def_p-GNWE}
\end{equation}
\hfill\QEDopen
\end{definition}  
\smallskip

We have defined both the game and the set of equilibria we are interested in. Let us now elaborate on the convergence properties of the game in \eqref{eq:game}, providing  three examples highlighting different aspects of these dynamics. By means of the first two examples, we show, first  that the dynamics in \eqref{eq:myopic_BR} can fail to converge to an equilibrium point, even in the case of a static communication network, where the existence of a GNWE is guaranteed by \cite[Prop.~4]{grammatico:18tcns}  and then that the existence of p--GNWE is not guaranteed. Finally, the last example shows a case where the game in \eqref{eq:game} converges.

\smallskip
\begin{example}[non--convergence]
\label{ex:non_converging}
Consider a 2-player constrained game, where, for $i\in\{1,2\}$,  $x_i \in \bR$ and the local feasible decision set is defined as $ \ca X_i(u) := \{ v \in \bR \, | \, u+v = 0 \} = \{ - u \} $ and \textit{does not} vary over time. The collective feasible decision set is convex and reads as $
\bld{\ca X}:= \{ \bld x \in \bR^2 \, |\, x_1 + x_2 = 0 \}
$, hence the game is jointly convex. 
The dynamics of the game are as in \eqref{eq:myopic_BR}, and can be rewritten in closed form as the discrete-time linear system:
\begin{align}
\begin{bmatrix}
x_1(k+1) \\ x_2(k+1)
\end{bmatrix}
=
\begin{bmatrix}
0 & -1\\
-1 & 0
\end{bmatrix}
\begin{bmatrix}
x_1(k) \\
x_2(k)
\end{bmatrix},
\end{align}
which is not globally convergent, e.g., consider $x_1(0)=x_2(0)=1$. \hfill\QEDopen
\end{example}
\smallskip
\begin{example}[equilibirum existence]\label{ex:existence_GNWE}
Consider a 2-player game without local or coupling constraints and scalar strategies. The communication network can vary between the two graphs described respectively by the adjacency matrices $A_1 =\left[\begin{smallmatrix}\nicefrac{1}{2} & \nicefrac{1}{2}\\
\nicefrac{1}{2} & \nicefrac{1}{2} \end{smallmatrix}\right]$ and  $A_2 =\left[\begin{smallmatrix}\nicefrac{1}{3} & \nicefrac{2}{3}\\
\nicefrac{1}{3} & \nicefrac{2}{3} \end{smallmatrix}\right]$. The cost functions of the agents are in the form of \eqref{eq:def_costFunction}, where the local part is chosen as $\bar f_i(x_i) = \frac{1}{2}\lVert x_i - i\rVert^2$, for $i\in\{1,2\}$. For each one of the communication networks, there exists only one equilibrium point of the game, i.e., $\bld x_{A_1} = [\nicefrac{5}{4}\,,\: \nicefrac{7}{4}]^\top$ and $\bld x_{A_2} = [\nicefrac{4}{3}\,, \:\nicefrac{11}{6}]^\top$, when respectively $A_1$ or $A_2$ is adopted. Therefore the set of p--GNWE of the game is empty, leading the dynamics to oscillate between $\bld x_{A_1}$ and $\bld x_{A_2}$. \hfill\QEDopen
\end{example}
\smallskip
\begin{example}[convergence]
\label{ex:convergence}
Once again, consider the a 2-player game, where for a player $i\in\{1,2\} $ the local feasible set is $\Omega_i=[-1,1]$ and $f_i(x_i) = 0$. The collective feasible decision set is defined as 
$$\bld{\ca X}(k) \coloneqq \{ \bld x \in [-1,1]^2\,|\, m(k)\leq x_1+x_2\}$$
where $m(k)\in[-1,-0.25]$.
We choose $A(k)$ satisfying Standing Assumption~\ref{ass:row_stoch} and it is doubly stochastic, for every time instant $k\in\bN$. If the strategy profile belongs to the consensus subspace $\ca C$, both agents achieve the minimum of their cost function, and therefore all those points are equilibria of the unconstrained game. Furthermore, for the set $\hat{\ca C} = \{u\in\bR^2\,|\, u=\alpha \1^\top,\: \alpha\in[-0.25,1] \}$, it always holds that $\hat{\ca C}\subseteq \ca C\cap \bld{\ca X}$, and hence they are p-GNWE of the game. Assume that at $\bar k>0$, $m(\bar k)=-0.25$, then, for all $k>\bar k$, the dynamics reduce to $\bld x(k+1)=A(k)\bld x(k)$, therefore the profile strategy will converge to a  point in $\hat{\ca C}$, i.e., to a p--GNWE of the game. \hfill\QEDopen
\end{example}

\subsection{Primal--dual characterization}
As illustrated in Example~\ref{ex:non_converging}, the myopic constrained dynamics in \eqref{eq:myopic_BR} can fail to  converge, and thus we recast them as pseudo collaborative ones. The idea is that each player will minimize its own cost function, while at the same time coordinate with the others to satisfy the constraints. With this approach, we aim to achieve asymptotic fulfillment of the coupling constraints. As a first step, we dualize the dynamics introducing, for each player $i\in\ca N$, a dual variable $\lambda_i\in\bR^M_{\geq 0}$. The arising problem is an \textit{auxiliary (extended) network game}, see  \cite[Ch.~3]{cominetti:facchinei:lasserre}. The collective vector of the dual variables is denoted by $\bld \lambda \coloneqq \col((\lambda_i)_{i\in\ca N})$. The equilibrium concept is adapted to this modification in the dynamics, so we define the \textit{persistent Extended Network Equilibrium} (p--ENWE).

\smallskip 
\begin{definition}[persistent Extended Network Equilibrium] \label{def:Extended_NE}
The pair $(\boldsymbol{\overline  x},\bld{\overline{\lambda}} )$,  is a p--ENWE for the game in~\eqref{eq:game} if there exists $\bar k>0$ such that, for every $ i\in\ca N$,
\begin{align}\nonumber
\overline x_i & = \bigcap_{k\geq\bar k}\argmin_{y\in\bR^n} \, J_i \left( y,\textstyle\sum_{j=1}^Na_{i,j}(k)\overline{x}_j \right) + \overline \lambda^\top_i C_i(k)y,\\
\label{eq:cond_ENE}
\overline \lambda_i & = \bigcap_{k\geq\bar k} \argmin_{\xi\in\bR^M_{\geq 0}} \, -\xi^\top (C(k)\boldsymbol{\overline x}-c(k)) \:.
\end{align}
\hfill\QEDopen
\end{definition}
\smallskip

In the following, we assume the presence of a \textit{central coordinator} facilitating the synchronization between agents. This approach aligns with the new pseudo-collaborative behaviors of the agents, which  is widely used in the literature. The central coordinator broadcasts an auxiliary variable $\sigma\in\bR^M$ to each agent $i$, that, in turn, uses this information to compute its local dual variable $\lambda_i$.  
Specifically, at every time instant $k$, the agent scales the received variable $\sigma(k)$, by a possibly time-varying factor $\alpha_i(k)\in[0,1]$, attaining  in this way its local dual variable, i.e., $\lambda_i(k) \coloneqq \alpha_i(k) \sigma(k)$. 
The scaling factors $\alpha_i$ describe how the burden of satisfying the constraints are divided between the agents, hence $\sum_{i=1}^N \alpha_i = 1$. If  $\alpha_i= \nicefrac{1}{N} $, for all $i\in\ca N$, then the effort to satisfy the couplying constraints is fairly splitted between the agents, this case is considered in several works, e.g.,~\cite{cenedese:belg:gra:cao:2019:Asynch_ECC,yi_pavel:2017:disribiuted_primal_dual_conf,belgioioso_grammatico:2017:semi_decentralized_NE_seeking}. This class of problems was introduced  for the first time in the seminal work by Rosen \cite{rosen:65}, where the author formulates the concept of \textit{normalized equilibrium}. We adapt this idea for the problem at hand,  defining the \textit{persistent normalized extended network equilibrium} (pn-ENWE).
\smallskip 
\begin{definition}[persistent normalized-ENWE] \label{def:pn-ENWE}
The pair $(\boldsymbol{\overline  x},\overline{\sigma } )$, is a pn--ENWE for the game in~\eqref{eq:game}, if it exists $\bar k>0$, such that for all $ i\in\ca N$ it satisfies
\begin{align}\nonumber
\overline x_i & = \bigcap_{k\geq\bar k} \argmin_{y\in\bR^n} \, J_i \left(y,\textstyle\sum_{j=1}^Na_{i,j}(k)\overline{x}_j \right) + \alpha_i(k)\,\overline \sigma^\top C_i(k)y,\\
\label{eq:cond_pnENE}
\overline \sigma & = \bigcap_{k\geq\bar k} \argmin_{\varsigma\in\bR^M_{\geq 0}} \, -\varsigma^\top (\bld C(k) \boldsymbol{\overline x}-c(k)),
\end{align}
with $\alpha_i(k)>0$.
\hfill\QEDopen
\end{definition}
\smallskip
The following lemma shows that a pn--ENWE is also a p--GNWE, and vice versa.
\smallskip
\begin{lemma}[p--GNWE as fixed point]\label{lemma:GNWE_as_fixed_point} The following statements are equivalent:
\begin{enumerate}[(i)]
\item $\overline{\boldsymbol x}$ is a  p--GNWE for the game in \eqref{eq:game};
\item $\exists \overline \sigma\in\bR^{M}$ and $\bar k>0$ such that $\col(\overline{\boldsymbol x},\overline \sigma ) \in \ca E$, where $\ca E$ is the set of all the pn--GNWE of the game \eqref{eq:game}.
\hfill\QEDopen
\end{enumerate}
\end{lemma}
\smallskip
We omit the demonstration of the lemma, since it is analogous to that in \cite[Lem.~2]{grammatico:18tcns}.

This reformulation of the problem addresses the criticism highlighted in Example~\ref{ex:existence_GNWE}. In the following, we develop a distributed iterative algorithm converging to a p-GNWE of the original game in~\eqref{eq:game}.

\begin{figure*}[ht]
\hrule
\begin{subequations}\label{eq:tv_prox-GNWE}
\begin{align}\label{eq:tv_prox-GNWE 1}
\forall i \in \ca N \,:\quad \tilde x_i  &= \prox_{ \textstyle{\frac{\delta_i(k)}{\delta_i(k)+1}} \bar f_{i}} \left(  \textstyle{\frac{\delta_i(k)}{\delta_i(k)+1}} \left( \textstyle{\frac{1}{\delta_i(k)}} x_i  + \textstyle{\sum_{j=1}^N} a_{i,j}(k) x_j - \alpha_i(k) C_i^\top(k) \sigma \right)   \right)\\ \label{eq:tv_prox-GNWE 2}
\tilde \sigma  &= \proj_{\bR_{\geq 0}^M} \left(\sigma +  \textstyle{\frac{1}{\beta(k)}} \left(C(k) \bld x  -c(k)\right)\right)\\ \label{eq:tv_prox-GNWE 3}
\forall i \in \ca N \,:\quad x_i^+ &= x_i + \gamma(k) q_i(k) \left[ \delta_i(k)(\tilde x_i-x_i) +\textstyle{\sum_{j=1}^N}a_{i,j}(k)(\tilde {x}_j - x_j)  - \alpha_i(k) C_i^\top(k)\left(\tilde \sigma - \sigma\right) \right]\\ \label{eq:tv_prox-GNWE 4}
\sigma^+ &= \sigma + \gamma(k)\big[\beta(k) (\tilde \sigma -\sigma ) + C(k)(\bld {\tilde x}-\bld x )\big]
\end{align}
\end{subequations}
\hrule
\end{figure*}

\subsection{On the existence of persistent equilibria}
\label{subsec:eq_existence}

We devote the remainder of the section to a more in depth analysis of the problem of the existence of a p--GNWE for the game in \eqref{eq:game}. In general, there is no guarantee that such an equilibrium exists, as shown in Example~\ref{ex:existence_GNWE}.
The literature dealing similar problems is split on how to handle this problem. Namely, two possible assumptions can be adopted to proceed with the analysis. The first one supposes \textit{a priori} the existence of at least one p--GNWE in the game. This assumption does not restrict the problem at hand, since the convergence can be established only for the cases in which it is satisfied. However, it can be difficult to check if this assumption holds in practice. This approach is the one chosen in this work and it is usually adopted when the focus is more on theoretical results, see \cite[Cor.~5.19]{bauschke:combettes}, \cite[Prop.~3.1]{combettes:yamada:15}, \cite[Ass.~3]{grammatico:18tcns} and \cite[Ass.~6]{cenedese:2019:arXiv}. 

\smallskip
\begin{stassumption}[Existence of a pn-ENWE]
\label{ass:exist_PGNWE}
The set of pn-ENWE of \eqref{eq:game} is non-empty, hence $\ca E \not= \varnothing$ .
\hfill\QEDopen
\end{stassumption}
\medskip

On the other hand, the second assumption considers only those games in which the $N$ local cost functions share at least one common fixed point. This implies that at least one point in the consensus subspace is an equilibrium invariant to the change of the communication network. If, at the same time, this point is also feasible, then it is a p--GNWE of the game. This assumption is clearly stronger than the previous one. Nevertheless, it is easier to verify in practice, since it only requires the analysis of the cost functions of the agents, as shown in Example~\ref{ex:convergence}.
Mainly for this reason, it is widely spread throughout the literature, where  it is either implicitly verified as in \cite{nedic:ozdaglar:parrillo:10} or explicitly required \cite[Ass in Th.~2]{fullmer:morse:2018:common_fixed_point_finite_family_paracontraction} . 
 
\section{Convergence result}
\label{sec:conv_result}
Next, we propose the main result of this paper, an iterative and decentralized algorithm converging to  a pn-GNWE of the game in \eqref{eq:game}. We call it TV--Prox--GNWE and it is reported in \eqref{eq:tv_prox-GNWE 1}--\eqref{eq:tv_prox-GNWE 4}, while its complete derivation is described in the Appendix.

In order to provide the bounds for the choices of the parameters in the algorithm, let us redefine the matrix $\bld A(k)$ via a diagonal matrix, an upper and a lower triangular matrix, i.e., $\bld A(k) = \bld A_{\textup{ut}}(k) + \bld A_{\textup{d}}(k) +\bld A_{\textup{lt}}(k)$, where $\bld A_{\textup{ut}}$ and  $\bld A_{\textup{lt}}$ always have zeros diagonal elements. For each time instant $k\in\bN$, the parameters in TV--Prox--GNWE  are set such that, the following inequalities hold:
\begin{subequations}
\label{eq:bound_param} 
\begin{align}\label{eq:bound_param1} 
&\min_{i\in\ca N}( \delta_i^{-1} + a_{i,i})  \geq \lVert \bld A-\bld A_{\textup{d}} \rVert + \lVert \bld C^\top - \bld{\Lambda C}^\top \rVert\\\label{eq:bound_param2} 
&\max_{i\in\ca N}( 2 q_i (\delta_i^{-1} +  a_{i,i}))  < R  +\gamma^{-1} \\\nonumber
&R\coloneqq 2\lVert \bld {QA}_{\textup{ut}} +  \bld {AQ}_{\textup{lt}}  \rVert +  \lVert \bld Q( \bld C^\top-\bld{\Lambda C}^\top ) \rVert\\\label{eq:bound_param3} 
&\beta  \geq {\textstyle \frac{1}{2}} \lVert \bld C - \bld{C \Lambda} \rVert \\\label{eq:bound_param4} 
&\beta  < {\textstyle \frac{1}{2}}\big(\gamma^{-1} - \lVert \bld C - \bld{C \Lambda} \rVert\big)
\end{align}
\end{subequations}
where $\bld \Lambda(k)\coloneqq\diag((\alpha_i(k))_{i\in\ca N})\otimes I_n$ and $\bld Q(k)\coloneqq \diag(( q_i(k) )_{i\in\ca N})\otimes I_n$, with $q_i$ being the $i$-th element of the left Perron-Frobenius eigenvector of $ A(k)$. Also in this case, we omitted the time dependency of the matrices to ease the notation.
The bounds in \eqref{eq:bound_param3}~--~\eqref{eq:bound_param4} implicitly lead to a condition on the maximum value of the step size $\gamma$, namely 
$\gamma \leq {\scriptstyle\frac{1}{2}} \lVert \bld C - \bld{C \Lambda}  \rVert^{-1}$.

The TV--Prox--GNWE in \eqref{eq:tv_prox-GNWE}, is composed of three main steps: a proximal gradient descend, performed by every agent \eqref{eq:tv_prox-GNWE 1},  a  dual ascend done by the central coordinator \eqref{eq:tv_prox-GNWE 2} and correction step, in \eqref{eq:tv_prox-GNWE 3}~--~\eqref{eq:tv_prox-GNWE 4}, to balance the asymmetricity of the weights  in the directed network, i.e., $a_{i,j}\not = a_{j,i}$.

The main technical result of the paper is the following theorem, where we  establish global convergence of the sequence generated by the TV--Prox--GNWE to a p-GNWE of the game in~\eqref{eq:game}.
\smallskip
\begin{theorem}
\label{th:TV_Prox_GNWE}
For all $i\in\ca N$ and $k\in\bN$, set $\alpha_i(k)= q_i(k)$, with $q_i(k)$ the $i$-th element of the left Perron--Frobenius eigenvector of $A(k)$, and choose $\delta_i(k)$, $\beta(k)$ and $\gamma$ satisfying \eqref{eq:bound_param}. Then, for any initial condition, the sequence $(\bld x(k))_{k\in\bN}$ generated by \eqref{eq:tv_prox-GNWE} converges to a p-GNWE of the game in  \eqref{eq:game}.~\hfill\QEDopen
\end{theorem}
\smallskip
\begin{proof}
See Appendix.
\end{proof}
\smallskip

\section{Simulation}
\label{sec:simulation}

In this section, we adopt TV--Prox--GNWE to solve a problem of constrained consensus. We consider a game with $N=15$ agents, where the strategy of every agent $i$ is $x_i\in\bR^5$, and its local feasible decision set is $\Omega_i\in[m_i,\,M_i]$, with $m_i$ and $M_i$ randomly drawn respectively from $[-100,-5]$ and $[5,100]$. The local cost function is equal to $f_i(x_i)=\iota_{\Omega_i}(x_i)$. The adjacency matrices, descibing the communication network at every time instant $k$, are randomly generated and define digraphs of the type small-word, satisfying Standing Assumption~\ref{ass:row_stoch}. The coupling constraints are used to force the strategies towards the consensus subspace and are in the form $| x_i(k) - x_j(k)| \leq s(k)\1$, for every $i,j\in\ca N$, where $s(k)>0$ and it is decreasing over time. Notice that in this case the \textit{multiplier graph} is complete, see \cite{yi_pavel:2017:disribiuted_primal_dual_conf}. Finally, the parameters of the algorithm are chosen such that they always satisfy~\eqref{eq:bound_param}. 

\begin{figure}[ht]
	\centering
	\includegraphics[scale=0.17]{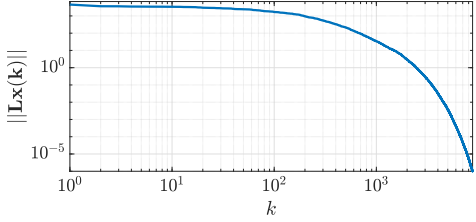}
	\caption{ Convergence of the strategy profile $\bld x(k)$ to the consensus subspace. The matrix $\bld L$ is the Laplacian matrix associated to the multiplier graph.}
	\label{fig:consensus}
\end{figure}
\begin{figure}[ht]
	\centering
	\includegraphics[scale=0.17 ]{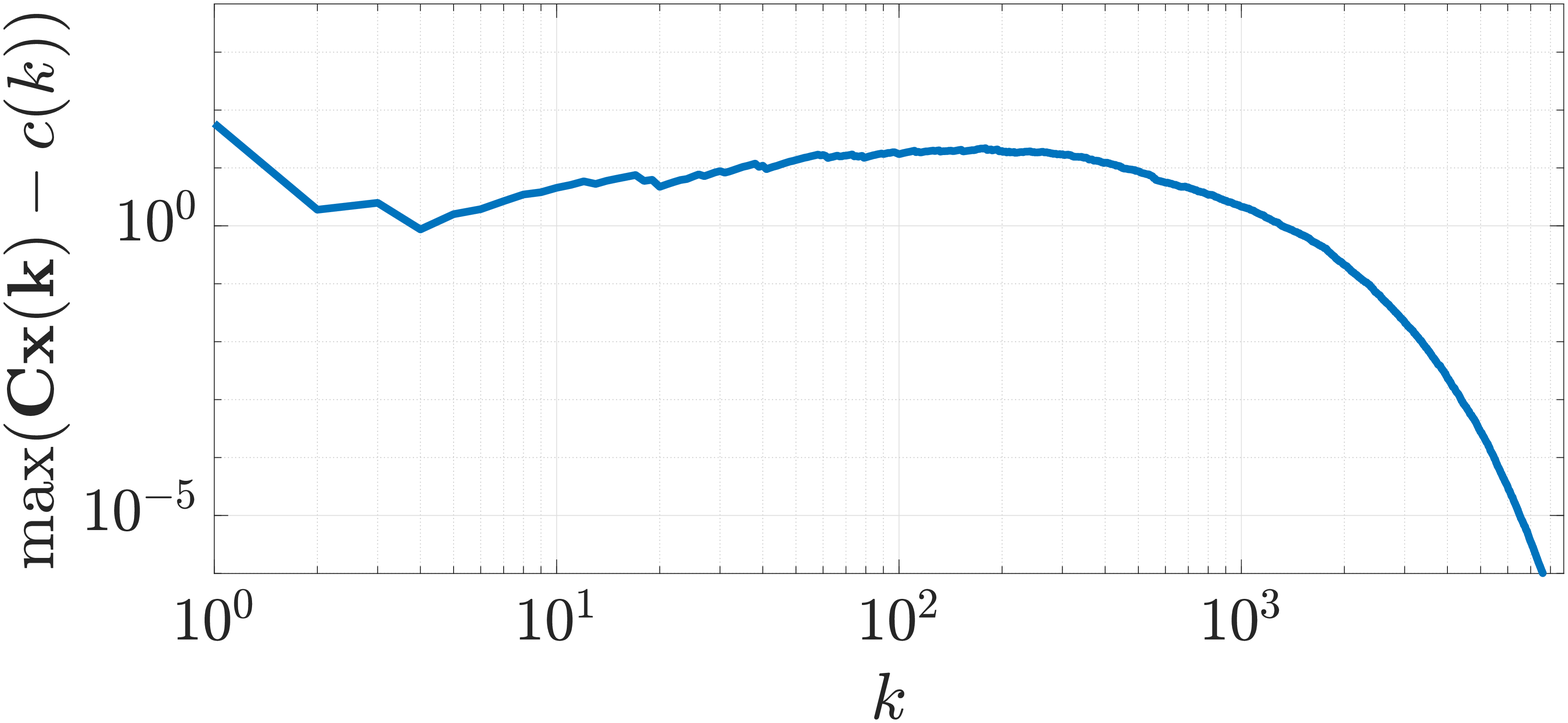}
	\caption{Asymptotic satisfaction of the time-varying affine coupling constraints $\bld{C x} (k)\leq c(k)$.}
	\label{fig:cnstr_viol}
\end{figure}
The trajectory of the profile strategy generated by TV--Prox--GNWE converges to the consensus subspace, this is shown in Fig.~\ref{fig:consensus}, by means  of the Laplacian matrix $\bld L$ of the multiplier graph. The initial strategy profile $\bld x(0)$ is randomly chosen in $\bld \Omega$. As expected from the result in Theorem~\ref{th:TV_Prox_GNWE}, the constraints are satisfied asymptotically, see Fig.~\ref{fig:cnstr_viol}.  

\section{Conclusion and outlook}
 \label{sec:conclusion}
In multi-agent network games, subject to time-varying coupling constraints and time-varying communication network, described by strongly connected digraphs, agents can fail to converge when they adopt proximal dynamics. 
Nevertheless, it is developed an iterative equilibrium seeking algorithms (TV--Prox--GNWE) that ensures the global convergence of the agents' strategies to an normalized equilibrium of the game, when it exists.

%
One of the most important open question in these type of problems regards the existence of an equilibrium point. 
This work can be improved with a new assumption for the equilibrium existence, which is general and easy to check. 
  
\appendix

\subsection{Algorithm derivation}
In this section, we propose the complete derivation of the iterative algorithm that we called  TV--Prox--GNWE. We divide the derivation in two mains steps
\begin{enumerate}
\item Equilibria reformulation
\item Modified proximal point algorithm
\end{enumerate}

\subsubsection{Equilibria reformulation}
 the set of pn-ENWE, defined by the two equalities in \eqref{eq:cond_pnENE}, can be equivalently rephrased as the set of fixed points of a suitable mappings.
First, we introduce the block-diagonal proximal operator 
\begin{equation} \label{eq:group_proximity_operator}
\boldsymbol{\prox_{f}} \left(
\begin{bmatrix}
z_1 \\
\vdots \\
z_N
\end{bmatrix}
\right) := 
\begin{bmatrix}
\textrm{prox}_{\bar f_1}(z_1) \\
\vdots\\
\textrm{prox}_{\bar f_N}(z_N)
\end{bmatrix}.
\end{equation}

In \eqref{eq:cond_pnENE}, the first equality is equivalent to 
$$\overline{\bld x} = \cap_{k>\bar k}\: \boldsymbol {\prox_{f}} (\boldsymbol A(k) \bld x - \bld\Lambda(k) \bld C^\top(k)\overline \sigma )\: ,$$ 
where $\bld A(k)\coloneqq A(k)\otimes I_n$ and $\bld \Lambda(k)=\diag((\alpha_i(k))_{i\in\ca N})\otimes I_n$.  The second equality holds true if and only if 
$\overline \sigma = \textrm{proj}_{\bR^M} (\overline{\sigma}+\bld C(k)\boldsymbol{\overline x}-c(k)) $.

In order to describe via operators these two relations, we define the static mappings 
\begin{equation}\label{eq:cal_F_map}
\boldsymbol{\ca R } :=\mathrm{diag} (\boldsymbol {\prox_{f}}, \:\textrm{proj}_{\bR^M_{\geq 0}})
\end{equation} 
and  the time-varying affine one $\boldsymbol{\ca G}_k:\bR^{nN+M}\rightarrow \bR^{nN+M}$ as
\begin{equation}\label{eq:ca_G_map}
\begin{split}
\boldsymbol{\ca G}_k(\cdot) &:=\boldsymbol{G} \,\cdot + \begin{bmatrix}
\boldsymbol 0\\c(k)
\end{bmatrix}\\
  &:= \begin{bmatrix}
\boldsymbol A(k) & -\bld \Lambda(k) \bld C^\top (k)\\
C(k) & I
\end{bmatrix}\cdot - \begin{bmatrix}
\boldsymbol 0\\c(k)
\end{bmatrix}\:.  
\end{split}
\end{equation}

As a result, the dynamics of the game result equal to
\begin{equation}
\label{eq:dyn_TV_constr}
\begin{bmatrix}
\bld x(k+1) \\
\sigma (k+1)
\end{bmatrix} = \bld{\ca R}\circ \bld{\ca G}_k\left(\begin{bmatrix}
\bld x(k) \\
\sigma (k)
\end{bmatrix} \right)\:.
\end{equation}

We exploit this new compact form to describe the set of pn--ENWE via the fixed points of $\bld{\ca R}\circ\boldsymbol{\ca G}_k$. In particular, by Definiton~\ref{def:pn-ENWE}, a pair $(\bar {\bld x}, \bar \sigma )$ is a pn--ENWE of the game in \eqref{eq:game} if and only if $\col((\bar {\bld x}, \bar \sigma ))\in\cap_{k>\bar k} \fix(\bld{\ca R}\circ\boldsymbol{\ca G}_k)$. Furthermore, from Lemma~\ref{lemma:GNWE_as_fixed_point}, we also know that a pn--ENWE is a p--GNWE of the original game. So, we focus on the design of an algorithm converging to the subset $\ca E$ for which we can take advantage of this new formulation.

A useful tool to solve fixed point seeking problem is to reformulate it as a zero finding problem, as done in the next lemma, see \cite[Ch.~26]{Bauschke2010:ConvexOptimization}.
\smallskip
\begin{lemma}[{\cite[Prop.~26.1 (iv)]{Bauschke2010:ConvexOptimization}}] \label{lemma:fixed_point_as_zero}
Let $\bld{\ca B} : = F \times N_{\bR^{M}_{\geq 0}}$, with $F:=\prod_{i=1}^N \partial \bar f_i$. Then,
$$\fix \left(\boldsymbol{\ca R } \circ \boldsymbol{\ca G }_k \right) = \textrm{zer}\left( \bld{\ca A}_k  \right)\:,$$ where $\bld{\ca A}_k\coloneqq\bld{\ca B} + \Id - \boldsymbol{\ca G }_k$.~\hfill\QEDopen
\end{lemma}
\smallskip

\subsubsection{Modified proximal point algorithm}
we describe in details the passages to develop the iterative  algorithm solving the zero finding problem associated to the operator $\bld{\ca A}_k$, and, as a consequence, the original one of finding pn--ENWE of \eqref{eq:game}. We adopt a modified version of the \textit{proximal point algorithm} (PPP) (see \cite[Prop.~23.39]{Bauschke2010:ConvexOptimization} for its standard formulation).
In particular, the update rule is a preconditioned version of the PPP algorithm proposed in \cite[Eq.~4.18]{briceno:davis:f_b_half_algorithm}, after defining $\bld \varpi:=\col(\bld x(k),\,\sigma (k))$ and $\bld \varpi^+:=\col(\bld x(k+1),\,\sigma (k+1))$, it can be rewritten as 
\smallskip
\begin{subequations}
\label{eq:PPP_update_rule}
\begin{align}
\label{eq:PPP_update_rule1}
\tilde{\bld \varpi} &=\mathrm{J}_{\Phi^{-1}(k) \bld{\ca A}_k } \bld \varpi \\
\label{eq:PPP_update_rule2}
\bld \varpi^+ &= \bld \varpi + \gamma(k) \bld{\bar Q}(k)\Phi(k) (\tilde{\bld \varpi}-\bld \varpi)
\end{align}  
\end{subequations}
\smallskip
where and $\gamma(k)>0$  is the step--size of the algorithm and $\bld{\bar Q}(k)\coloneqq \diag(\bld Q(k),I)$. 
The preconditioning matrix is chosen as
\smallskip
\begin{equation}
\label{eq:preconditioning_matrix}
\Phi(k) := \begin{bmatrix}
\bld{\delta}^{-1}(k)+\bld A(k) & -\bld \Lambda(k) C(k)^\top\\
C(k) & \beta(k) I_M 
\end{bmatrix}
\end{equation}
\smallskip
where  $\beta(k)\in\bR_{>0}$ and $\bld \delta(k)\coloneqq \diag((\delta_i(k))_{i\in\ca N})\otimes I_n$. The \textit{self-adjoint} and \textit{skew symmetric} components are defined as $U(k) := (\Phi(k)+\Phi^\top(k))/2$ and $S(k):=(\Phi(k)-\Phi^\top(k))/2$. Due to the non symmetric preconditioning the resolvent operator takes the form
$$\mathrm{J}_{\Phi^{-1}(k) \bld{\ca A}_k }:=\mathrm{J}_{U^{-1}(k) (\bld{\ca A}_k +S(k))}(\Id + U^{-1}(k)S(k))\:.$$

The parameters $\bld \delta(k)$ and $\beta(k)$ in the preconditioning have to be chosen such that $U(k)\succ 0 $ and $\lVert \bld{\bar Q}(k) U(k) \rVert\leq \gamma^{-1}(k)$. This can be done via the \textit{Gerschgorin Circle Theorem} for partitioned matrices (more stringent but more involved bounds can be obtained via \cite[Th.~2.1]{melman:2010:generalization_gershgorin_disk}). The resulting bounds are reported in \eqref{eq:bound_param}.

Using a reasoning akin to the one in \cite[Proof of Th.~4.2]{briceno:davis:f_b_half_algorithm}, one can show that, at every time instant $k$, the set of fixed points of the mapping describing the update in  \eqref{eq:PPP_update_rule1}--\eqref{eq:PPP_update_rule2} 
coincides with $\zer(\bld{\ca A}_k)$.

Finally, we are ready for the complete derivation of the algorithm by explicitly compute the local update rules of the agents and of the central coordinator. We omit the time dependency in the following formulas.

First, we focus on  \eqref{eq:PPP_update_rule1} , so
\begin{subequations}
\begin{align}
\tilde{\bld \varpi} = \mathrm{J}_{U^{-1} (\bld{\ca A} +S)}(\Id + U^{-1}S) \bld \varpi \\
\tilde{\bld \varpi} + U^{-1} (\bld{\ca A} +S)\tilde{\bld \varpi} \ni \bld \varpi + U^{-1}S\bld \varpi\\
\bld 0 \in U(\tilde{\bld \varpi}- \bld \varpi) + \bld{\ca A}\tilde{\bld \varpi} +S(\tilde{\bld \varpi}- \bld \varpi)\\\label{eq:PPP_update_expr}
\bld 0 \in \Phi(\tilde{\bld \varpi}- \bld \varpi) + \bld{\ca A}\tilde{\bld \varpi} 
\end{align}
\end{subequations}  
  
By solving the first row block of \eqref{eq:PPP_update_expr}, i.e. $\bld 0 \in (\bld \delta^{-1} +\bld A)(\tilde{ \bld x} - \bld x) - \bld \Lambda C^\top(\tilde \sigma-\sigma) +F(\tilde{ \bld x})+\tilde{ \bld x} -\bld A \tilde{ \bld x} + \bld \Lambda C^\top \tilde \sigma  $ , we obtain
\begin{align*}
\bld 0_{nN} &\in \bld \delta^{-1}(\tilde{ \bld x} - \bld x) -\bld A \bld x + \bld \Lambda C^\top\sigma +F(\tilde{ \bld x})+\tilde{ \bld x}\\
\bld 0_{nN} &\in (\bld \delta^{-1}+I)\tilde{ \bld x} +F(\tilde{ \bld x}) - \bld \delta^{-1}\bld x -\bld A \bld x + \bld \Lambda C^\top\sigma \:.
\end{align*} 
Let us define, with a small abuse of notation, the matrix $\frac{1}{ \delta^{-1}+1}\coloneqq \diag\left(\big(\frac{1}{\bld \delta_i^{-1}+1}\big)_{i\in\ca N}\right)\otimes I_n$,
then we attain
\begin{align}\nonumber
\bld 0_{nN} &\in \tilde{ \bld x} +\textstyle{\frac{1}{\bld \delta^{-1}+1}} F(\tilde{ \bld x}) +\textstyle{\frac{1}{\bld \delta^{-1}+1}}\big[ \bld \Lambda C^\top\sigma - \bld \delta^{-1}\bld x -\bld A \bld x  \big]\\ \label{eq:first_row_block}
 \tilde{ \bld x} &=  \mathrm{J}_{\scriptstyle{\frac{1}{\bld \delta^{-1}+1}} F} \left( \textstyle{\frac{1}{\bld \delta^{-1}+1}}\big[ \bld \delta^{-1}\bld x +\bld A \bld x  -\bld \Lambda C^\top\sigma \big]  \right)\:.
\end{align}

The second row block instead reads as $\bld 0 \in C (\tilde{ \bld x} - \bld x) +\beta (\tilde \sigma-\sigma) +N_{\bR^{M}_{\geq 0}}(\tilde \sigma) + \tilde \sigma - C \tilde{ \bld x} +c $, and leads to 
\begin{subequations}
\begin{align}\label{eq:second_row_block 1}
\bld 0_{M} &\in -C\bld x +\beta (\tilde \sigma-\sigma) +N_{\bR^{M}_{\geq 0}}(\tilde \sigma) +c  \\\label{eq:second_row_block 2}
\tilde \sigma &= \mathrm{J}_{N_{\bR^{M}_{\geq 0}}}\left( \sigma+\textstyle{ \frac{1}{\beta} } (C\bld x -c)  \right)
\end{align}
\end{subequations}

Combining \eqref{eq:first_row_block} and \eqref{eq:second_row_block 2} together with \eqref{eq:PPP_update_rule2}, leads to the final formulation for the  TV--Prox--GNWE, its dynamics are shown in \eqref{eq:tv_prox-GNWE 1}~--~\eqref{eq:tv_prox-GNWE 4}.

\subsection{Convergence proof of TV--Prox--GNWE}
\label{app:TV_Prox_GNWE}

In order to simplify the proofs proposed in the following, let us introduce some useful definition that will be adopted thorough the whole section. 
 We define the two scalars $L_k$ and $m_k$, the former is the Lipschitz constant of $S(k)$ and the latter is such that $m_k\lVert x \rVert^2<\langle U(k)x,x\rangle$, this also implies that $\lVert U^{-1}(k) \rVert\leq m_k^{-1}$.
Next, we define the time-varying matrix $K(k)\coloneqq \overline{\bld Q}_k U(k)$  and  the scalars $\rho \coloneqq m_k^{-1}L_k$, $q_m \coloneqq \min_{i} [\bld{\overline{Q}}_k]_{ii}$ and $M_k\geq \lVert U(k) \rVert$. Notice that, without loss of generality, we can always choose the normalized version of the left Perron Frobenius eigenvector $q(k)$ of the matrix $A(k)$, so $ \max_{i} [\bld{\overline{Q}}_k]_{ii} \leq 1$.

In the following, we also omit the time dependency of the operators when this does not lead to ambiguities. The proofs follow similar steps to the ones in \cite[Prop.2.1~and~4.2]{briceno:davis:f_b_half_algorithm}, where the case of a static communication network is considered.
\smallskip

\begin{lemma}
\label{lem:QNE_operator}
For all $k\in\bN$, consider  the time-varying operator \begin{equation}
\label{eq:T_Phi_def}
\ca T_{\Phi}:= \mathrm{J}_{\Phi^{-1}\bld{\ca A}}+U^{-1}S(\mathrm{J}_{\Phi^{-1}\bld{\ca A}}-\Id)\:,
\end{equation}
then the following hold: 
\begin{enumerate}[(i)]
\item $\ca T_{\Phi}$is quasi-nonexpansive in the space $\ca H_{K}$,
\item if $L_k\leq m_k $, then $\fix(\ca T_{\Phi}) = \zer(\bld{\ca A})$ .\hfill\QEDopen
\end{enumerate} 
\end{lemma}
\begin{proof}
(i) From the proof of \cite[Th.~5]{cenedese:2019:arXiv}, the operator $\bld{\ca A} \coloneqq  \bld{\ca B}+\Id-\bld{\ca G}_k$ is maximally monotone in $\ca H_{\overline{\bld Q}}$ for all $k\in\bN$. The operator $S$ is also maximally monotone, since it is a skew-symmetric matrix, \cite[Ex.~20.30]{Bauschke2010:ConvexOptimization}. Next, we define the two auxiliary operators $\ca M = U^{-1}(\bld{\ca A}+S)$ and $\ca D = - U^{-1}S$. It is easy to see that $\ca A$ is monotone and $\ca D$ is monotone and  $\rho$-Lipschitz, , in $\ca H_K$. From \cite[Prop.~2.1~(3)]{briceno:davis:f_b_half_algorithm}, we obtain that for $\varpi^*\in \fix(\ca T_\Phi)$ it holds for some $\bar\gamma>0$
\begin{equation}
\label{eq:T_phi_QNE}
\lVert \ca T_{\Phi} \varpi -\varpi^* \rVert^2_{K} \leq \lVert \varpi-\varpi^* \rVert_{K} -(1+\rho^2) \lVert \varpi - x \rVert^2_{K}
\end{equation}
where $x \coloneqq \mathrm{J}_{\ca M} (\varpi- \ca D \varpi) $. Hence, we conclude that $\ca T_\Phi$ is quasi-nonexpansive in $\ca H_K$.
\smallskip 

(ii) If $L_k\leq m_k $, then the result follows directly from  \cite[Prop.~2.1~(1)]{briceno:davis:f_b_half_algorithm}, where we considered $A\coloneqq \ca M$, $B_1 = 0$ and $B_2\coloneqq \ca D$. 
\end{proof}
\subsection*{Proof of Theorem~\ref{th:TV_Prox_GNWE}}
 Consider $\ca T_{\Phi}$, as in \eqref{eq:T_Phi_def}, then the following relation always holds
\begin{equation}
\label{eq:relation_T_J}
\overline{\bld Q} U(\Id - \ca T_{\Phi} )(x) = \overline{\bld Q} \Phi ( x - \mathrm{J}_{\Phi^{-1} \bld{\ca A}})\, .
\end{equation}  
From Lemma~\ref{lem:QNE_operator} we know that $\ca T_{\Phi}$ is quasi-nonexpansive in $\ca H_{K}$, hence $\ca S:=(\Id+\ca T_{\Phi})/2$ belongs to the class $\mathfrak{I}$ in $\ca H_{K}$, \cite[Prop.~2.2(v)]{combettes:01}. From \cite[Prop.~4.1]{briceno:davis:f_b_half_algorithm} and \eqref{eq:relation_T_J}, we define the operator
\begin{equation}
\label{eq:W_def}
\begin{split}
\ca W_{\Phi} &: = \Id-\lVert K \rVert^{-1} K(\Id -\ca S) \\
&= \Id-\frac{1}{2}\lVert K \rVert^{-1} K (\Id -\ca T_{\Phi}) \: , \\ 
\end{split}
\end{equation}
belonging to $\mathfrak{ I}$ in $\ca H$ and $\fix(\ca W_{\Phi} ) = \fix(\ca T_{\Phi})=\zer(\bld{\ca A})$, where the last equality comes from Lemma~\ref{lem:QNE_operator}. The update rule in \eqref{eq:PPP_update_rule} can be rewritten as
\begin{equation}
\bld{\varpi}^+ = \bld{\varpi}+ 2 \gamma\lVert K \rVert (\ca W_{\Phi} \bld{\varpi} - \bld{\varpi} )\:,
\end{equation} 
where $\gamma \lVert K \rVert<1$ for all $k$, due to the choice of $\delta_i$ and $\beta$. Thus, from \cite[Th.~4.2(ii)~and~4.3]{combettes:01}, we have that $(\lVert \bld{\varpi}  - \ca W_{\Phi} \bld{\varpi} \rVert^2)_{k\in\bN}$ is summable and converges in $\ca H$ to an element $\overline{\bld{\varpi} }
\in \ca E_c$, hence to a pn-ENWE, if and only if every sequential cluster point of sequence belong to $\ca E_c$.  

Toward this aim, we define $x_\varpi \coloneqq \mathrm{J}_{\Phi^{-1} \bld{\ca A}} (\varpi)$ then we prove that $\varpi-x_\varpi \rightarrow 0$ and $x_\varpi \rightarrow \overline{x}_\varpi\in\zer(\bld{\ca A}) $, concluding the proof.

First, notice that, due to the choice of the coefficients $\bld \delta$ and $\beta$  in $\Phi(k)$, it holds 
\begin{equation}
\label{eq:conv varpi - T_phi}
\begin{split}
\lVert \bld{\varpi}  - \ca T_{\Phi} \bld{\varpi} \rVert^2_{K} & \leq    \lVert K^{-1}\rVert  \lVert  \bld{\varpi}  - \ca T_{\Phi} \bld{\varpi} \rVert ^2 \\
&\leq  4 m_k^{-1} M_k^2 \lVert  \bld{\varpi}  - \ca W_{\Phi} \bld{\varpi} \rVert ^2\rightarrow 0\:. 
\end{split}
\end{equation}
Let us define $x_\varpi \coloneqq \mathrm{J}_{\bar\gamma \ca A} (\varpi-\bar \gamma U^{-1} S \varpi)$,
Moreover, from \eqref{eq:T_phi_QNE}, it follows 
\begin{equation*}
\begin{split}
(1-\rho^2) \lVert \varpi&-x_{\varpi} \rVert_{K}  \leq \lVert \varpi -\varpi^*\rVert_{K} - \lVert \ca T_{\Phi} \varpi -\varpi\rVert_{K} \\
&= - \lVert \ca T_{\Phi} \varpi -\varpi\rVert_{K} -2\langle \ca T_{\Phi} \varpi -\varpi, \varpi -\varpi^* \rangle_{K}\\
&\leq - \lVert \ca T_{\Phi} \varpi -\varpi\rVert_{K} -2M_k\lVert \ca T_{\Phi} \varpi -\varpi\lVert_{K} \rVert \varpi -\varpi^* \rVert
\end{split}
\end{equation*}
The above inequality leads to 
\begin{align}
\label{eq:conv varpi-x_varpi}
(1-\rho^2) m_k q_{m} \lVert \varpi-x_{\varpi} \rVert &\leq  - \lVert \ca T_{\Phi} \varpi -\varpi\rVert_{K}\\ \nonumber
& -2M_k\lVert \ca T_{\Phi} \varpi -\varpi\lVert_{K} \rVert \varpi -\varpi^* \rVert 
\end{align} 
The sequence $(\lVert \varpi(k) -\varpi^*  \rVert )_{k\in\bN}$ is bounded and from \eqref{eq:conv varpi - T_phi} and  \eqref{eq:conv varpi-x_varpi}, we deduce that $\varpi-x_{\varpi} \rightarrow 0$. 
Notice also that 
$$ \lVert \overline{Q}\Phi (\varpi-x_{\varpi})\rVert \leq (M_k+L) \lVert  \varpi-x_{\varpi} \rVert \rightarrow 0\:.$$
From the definition of $x_\varpi$, it follows
$$ u_k \coloneqq \overline{Q}\Phi( \varpi - x_\varpi ) \in \bld{\ca A} x_\omega  \: . $$
Thus, since $u_k\rightarrow 0$, we conclude that $x_\varpi \rightarrow \overline{x}_\varpi\in\zer(\bld{\ca A}) $.

\hfill\QED


\bibliographystyle{IEEEtran}
\bibliography{diary_bibliography,librarySG}
\balance
 
\end{document}